\documentclass[Journal]{IEEEtran}

\usepackage[tracking=true]{microtype}

\newcommand\narrowstyle{\SetTracking{encoding=*}{-50}\lsstyle}

\usepackage{algorithmicx}
\usepackage{algorithm}
\usepackage{algpseudocode}

\usepackage{epsfig}
\usepackage{multirow}

\usepackage{amsthm}
\usepackage{amssymb}
\usepackage{amsmath}
\usepackage{bm}
\usepackage{cite}
\usepackage{breqn}
\usepackage{float}
\usepackage{mathrsfs}
\usepackage{mathtools}
\usepackage{amsfonts}
\usepackage{graphicx}
\usepackage{algorithm, algpseudocode}
\usepackage{subfigure}

\usepackage[usenames]{color}
\usepackage{xcolor}
\usepackage{amsfonts}
\usepackage{relsize}
\usepackage{lmodern}
\usepackage{slantsc}
\usepackage[mathscr]{eucal}
\usepackage{enumerate}
\usepackage[final]{pdfpages}
\usepackage{stfloats}
\usepackage{blindtext}

\usepackage{caption}
\usepackage{balance}
\usepackage{flushend}
\usepackage{epsfig}

\usepackage{multirow}

\usepackage{amsthm}
\usepackage{amssymb}
\usepackage{amsmath}
\usepackage{bm}

\usepackage{cite}

\usepackage{breqn}
\usepackage{float}
\usepackage{mathrsfs}
\usepackage{mathtools}
\usepackage{amsfonts}
\usepackage{graphicx}

\usepackage{algorithm, algpseudocode}

\usepackage{subfigure}
\usepackage[usenames]{color}
\usepackage{xcolor}
\usepackage{amsfonts}
\usepackage{relsize}
\usepackage{lmodern}
\usepackage{slantsc}
\usepackage[mathscr]{eucal}

\usepackage{dsfont}
\usepackage{bbm}
\usepackage{bm}

\usepackage{caption}
\usepackage{tabularx, booktabs}
\usepackage{adjustbox}

\usepackage{xcolor}

\usepackage{colortbl}
\usepackage{lipsum}
\usepackage{scalerel}
\usepackage{stackengine}
\usepackage{tabularx,ragged2e}
\usepackage{booktabs}

\usepackage[utf8]{inputenc}
\usepackage{array}
\usepackage{makecell}

\usepackage{enumitem}
\usepackage{makecell}
\usepackage{multicol}
\usepackage{scalerel}

\usepackage[font={small}]{caption}

\usepackage{accents}

\DeclareMathAlphabet{\mathpzc}{OT1}{pzc}{m}{it}
\DeclarePairedDelimiterX{\norm}[1]{\lVert}{\rVert}{#1}

\newtheorem{assumption}{Assumption}
\newtheorem{thm}{Theorem}

\algnewcommand\Input{\item[\hspace{6pt}\textbf{Input:}]}
\algnewcommand\Output{\item[\hspace{6pt}\textbf{Output:}]}
\algnewcommand\OutputVal{\textbf{output} }
\newcolumntype{L}[1]{>{\raggedright\arraybackslash}p{#1}}
\newcolumntype{C}[1]{>{\centering\arraybackslash}p{#1}}
\newcolumntype{R}[1]{>{\raggedleft\arraybackslash}p{#1}}

\begin{document} 
	\title{ \narrowstyle Tracking Consensus of Networked Random Nonlinear Multi-agent Systems with Intermittent Communications}   
	\author{Ali Azarbahram
		\thanks{Ali Azarbahram is with the Department of Mechanical and Process Engineering, Rheinland-Pfälzische Technische Universität
			Kaiserslautern-Landau, 67663 Kaiserslautern, Germany. Current contact via email address: ali.azarbahram@mv.uni-kl.de.} 	
	}

	
	\pagestyle{empty}   
	\maketitle
	\thispagestyle{empty}  
	\begin{abstract} 
		The paper proposes an intermittent communication mechanism for the tracking consensus of high-order nonlinear multi-agent systems (MASs) surrounded by random disturbances. Each collaborating agent is described by a class of high-order nonlinear uncertain strict-feedback dynamics which is disturbed by a wide stationary process representing the external noise. The resiliency level of this networked control system (NCS) to the failures of physical devices or unreliability of communication channels is analyzed by introducing a linear auxiliary trajectory of the system. More precisely, the unreliability of communication channels sometimes makes an agent incapable of sensing the local information or receiving it from neighboring nodes. Therefore, an intermittent communication scheme is proposed among the follower agents as a consequence of employing the linear auxiliary dynamics. The closed-loop networked system signals are proved to be noise-to-state practically stable in probability (NSpS-P). It has been justified that each agent follows the trajectory of the corresponding local auxiliary virtual system practically in probability. The simulation experiments finally quantify the effectiveness of our proposed approach in terms of providing a resilient performance against unreliability of communication channels and reaching the tracking consensus.
	\end{abstract}
	
	\IEEEpeerreviewmaketitle
	\section{Introduction}
	\label{Sec. 1} 
 
	\IEEEPARstart{T}{he} cooperative control of multi-agent systems (MASs) has certainly been among the most referred topics in the field of systems engineering in the past decade \cite{olfati2007consensus,azarbahram2022event,shahvali2020bipartite,amini2016h_,shahvali2023adaptive}. The control of collaborative MASs is inevitably linked with the concept of networked control systems (NCSs) \cite{7993094,azarbahram2022secure,amini2021sampled,amini2019resilient}. For instance, when a networked group of interconnected agents performs a specific mission, the resiliency of network to measurement failures in transmission channels is considered among the essential concerns that should be taken into account. Moreover, the real-world MASs are usually surrounded by environmental disturbances and the nature of these external forces is in fact random. The irrefutable presence of environmental disturbances has highlighted the significance of investigating and analyzing the stability of nonlinear MASs modeled by stochastic \cite{8709691,azarbahram2022eventvcvcvcvc,azarbahram2022platoon,8734885} or random \cite{WU2019314,azarbahram2022event} differential equations. With this general overview of the multi-agent NCSs challenges, now we review some of the most important technical considerations regarding the implementation of random networked MASs in more details.
\\
\indent
In real MAS applications, the induced disturbances by different factors certainly degrade the overall performance of system and in some cases cause instability and even system damage. In prior methods for controlling nonlinear MASs, time varying deterministic external disturbances are considered and compensated (see for instance, \cite{LI2021380}). However, environmental disturbances contain stochastic components. The stability analysis and control of nonlinear MASs modeled by stochastic differential equations including standard Wiener process is widely studied in recent years \cite{8709691, 8734885}. In theoretical approaches, white noise originated possess are favorable due to their convenience in analysis. However, such a process fails to fully model many realistic situations. Besides, a standard Wiener process is not strictly differentiable. Accordingly, colored noise processes have attracted much attention in recent years \cite{WU2019314, YAO2020108994, 9428607, wang2021distributed, bao2021containment}. Although the cooperative control of nonlinear MASs subject to colored noise disturbances is recently addressed in \cite{wang2021distributed} and \cite{bao2021containment}, agents sometimes cannot well sense the local information or receive it from neighboring nodes due to the failure of physical devices or unreliability of communication channels in real NCS implementations. Therefore, it is reasonable to assume that each agent can make the measurements only intermittently and the control system is supposed to compensate this effect \cite{li2007stabilization}. More precisely, investigating the cooperative performance of MASs subject to the wide stationary processes (which stands for the colored noise) is more favorable to be addressed while other NCSs limitations have also been taken into account.
\\
\indent
The intermittent control approach is known as a discontinuous mechanism that is applied periodically on the system \cite{guo2018distributed, li2020neural}. This control method has been widely studied in recent years and showed significant superiority in terms of increasing the flexibility and being cost-efficient compared to traditional continuous control schemes \cite{li2007stabilization}. The intermittent control has different practical applications in various real-world scenarios including a single system \cite{xavier2020practical, wang2021time, yang2021non, liu2020stabilization} and collaborative MASs \cite{guo2018distributed, li2020neural}. The state-dependent intermittent control approaches are proposed in \cite{xavier2020practical,wang2021time} for general classes of linear networked system. Howbeit, the nature of almost all real systems are intrinsically nonlinear. The intermittent control schemes are designed for single nonlinear networked systems in \cite{yang2021non, liu2020stabilization}. Subsequently, the collaborative performance of nonlinear MASs is recently studied in \cite{guo2018distributed, li2020neural} by introducing the intermittent mechanisms. However, the proposed approaches in \cite{guo2018distributed, li2020neural} are just valid for the networks of general second-order and first-order nonlinear MASs, respectively. Additionally, the environmental disturbances are not taken into account in the nonlinear dynamics of studied MASs in \cite{guo2018distributed, li2020neural}. 
\\
\indent
By the above discussion, this paper proposes an intermittent communication approach for the tracking consensus of uncertain high-order nonlinear MASs subject to environmental loads described by wide stationary process. According to the authors' best knowledge, this is the first instance that an intermittent data transmission mechanism is introduced for a class of strict-feedback MASs surrounded by random disturbance. The main contributions of this work are clarified in what follows. \textbf{(\textit{i})} Compared to the cooperative control designs in \cite{wang2021distributed} and \cite{bao2021containment} for nonlinear MASs subject to environmental loads that is described wide stationary process, the tracking consensus method in this paper is resilient to the failure of physical devices or unreliability of communication channels to better address the concerns of real NCS implementations. \textbf{(\textit{ii})} In comparison to the studied intermittent control approaches in \cite{guo2018distributed} and \cite{li2020neural}, the agents dynamics are described in this paper by uncertain high-order strict-feedback nonlinear random differential equations including wide stationary process to cover more real scenarios.  
\\
\indent
The remaining paper is organized as follows. The preliminaries and problem formulation is first given in section II. By modeling the intermittent communications over the network and stating that the objectives should be realized in two steps, the control architecture is constructed to meet the desired performance and then the closed-loop system stability is analyzed. In section IV, the effectiveness of the proposed method using a numerical example is investigated. The last section finally concludes the paper.

\section{Preliminaries and Problem Formulation}
\label{Sec. 2}

Throughout this paper, vectors and matrices are denoted in the bold font and scalars by normal font. For an undirected graph of $N$ followers, ${\bm{\mathcal{A}}}:=[a_{i,j}] \in \mathbb{R}^{N \times N}$ is the follower adjacency (or connectivity) matrix with element $a_{i,j}$. $ {N}_i $ is the neighbouring set for node $ i $. If follower $ i $ can directly obtain the information from the leader, then $ b_i > 0 $; otherwise $ b_i = 0 $, and ${\bm{\mathcal{B}}}:=diag[b_1,...,b_N]$ $\in \mathbb{R}^{N \times N}$. Furthermore, the Laplacian matrix is $ \bm{\mathcal{L}} $. \\
The following notations are also used. $\mathbb{R}$: the set of real scalars; $ \mathbb{N} $: the set of natural number; $\mathbb{R}^n$: the set of real column vectors with dimension $n$; $|x|$: absolute value of $x \in \mathbb{R}$; $||\bm{x}||$: Euclidean norm of the column vector $\bm{x}\in \mathbb{R}^n$; $ {\rm T} $: the transpose of vector or square matrix; $ \otimes $: the Kronecker product; $ \bm{1}_N $: column vector of order $ N $ with one entries.
\\
\indent
{\textit{Random Dynamics}:} A random differential equation which typically describes a system in the presence of colored noise is described as 
\begin{align} \label{equpjj1fff}   
\dot{x}= {h}({x},t) + g({x},t) \xi(t),
\end{align}
in which scalar $ x $ is the state and $ \xi(t) \in \mathbb{R}$ is a wide stationary process that represents the colored noise and stands for a noise with uneven power spectral density function. The nonlinear terms are defined by an unknown function $ {h}({x},t) $ and, a known function $ g({x},t) $ that represents the coefficient of the colored noise.
\\
\indent
\textbf{Definition 1\cite{WU2019314}:} Let the continuous zero-mean process $ \xi(t) $ is $ \mathcal{F}_t $-adapted  and satisfies $ \sup_{t \geq t_0} E_p (|\xi(t)|^2) \leq  \xi^{\ast} $ for a positive constant $ \xi^{\ast} $. The described random system \eqref{equpjj1fff} is called noise-to-state practically stable in probability (NSpS-P), if for any positive constant $ l_0 $ there exists a class-$ \scriptstyle \mathcal{K} \mathcal{L}$ functions $ r(\bullet,\bullet) $ and class-$ \scriptstyle \mathcal{K}$ function $ d(\bullet) $ such that for any initial condition $ x_0 $ and for all $ t \in [t_0~~\infty) $, one has 
\begin{align} \label{eq1}
P( |x(t)| \leq r(|x_0|, t-t_0) + d(\varpi)   ) \geq 1 - l_0
\end{align}
where $ \varpi = l^{[a]} \xi^{\ast} + l^{[b]}$ with $ l^{[a]} > 0 $ and $ l^{[b]} \geq 0 $. 
\\
\indent
\textbf{Lemma 1\cite{WU2019314}:} Suppose that there exists a function $ \mathcal{V} \in C^1$ and positive constants $  c_{\gamma}, a_1, a_2, l^{[a]}, l^{[b]} $. The random system \eqref{equpjj1fff} is said to have an NSpS-P unique global solution if 
\begin{align} \label{eq1iiii}
\begin{cases}
a_1 |x|^2 \leq   \mathcal{V}  \leq  a_2 |x|^2,  \\
\dot{\mathcal{V}} \leq - c_{\gamma} \mathcal{V}  + l^{[a]} |\xi|^2 + l^{[b]}.
\end{cases}
\end{align}
Furthermore if \eqref{eq1iiii} holds, then $ \lim_{t \rightarrow \infty } E_p (\mathcal{V} ) \leq \frac{\varpi}{c_{\gamma}}$.

\subsubsection{Problem formulation}
A network of $ N $ uncertain nonlinear follower agents labeled with $ i, i = 1, \ldots, N $ are supposed to reach the racking consensus with respect to the virtual leader signal $ z_{r} \in \mathbb{R}$. Let $ \bm{x}_{i} \in \mathbb{R}^n $ with $  \bm{x}_{i} = [x_{i,1},\thinspace x_{i,2}, \ldots, \thinspace x_{i,n}]^{\rm {T}} $  denote the state vector for the $i$-th agent. The dynamics of each uncertain follower agent subject to random disturbance is described by the following high-order nonlinear differential equations

\begin{align} \label{dyn}
\begin{cases}
\dot{x}_{i,q} = x_{i,q+1} + f_{i,q}(\bar{\bm{x}}_{i,q}) + \bm{g}_{i,q}^{\rm {T}}(\bar{\bm{x}}_{i,q}) \bm{\xi}_i,   \\
~~~~~ q=1,\ldots,n-1, \\
\dot{x}_{i,n} = u_i + f_{i,n}(\bar{\bm{x}}_{i,n}) + \bm{g}_{i,n}^{\rm {T}}(\bar{\bm{x}}_{i,n}) \bm{\xi}_i, \\
z_i = x_{i,1},
\end{cases}
\end{align}
where $ u_i $ and $ z_i $ are respectively the control input and the output of the $ i $-th agent. $ f_{i,q}(\bar{\bm{x}}_{i,q}): \mathbb{R}^q \rightarrow \mathbb{R} $ is an unknown smooth function with $  \bar{\bm{x}}_{i,q} = [x_{i,1},\thinspace x_{i,2}, \ldots, \thinspace x_{i,q}]^{\rm {T}}  $. $ \bm{\xi}_i \in \mathbb{R}^m $ is a continuous zero-mean and $ \mathcal{F}_t $-adapted process representing the colored noise that satisfies 
$ \sup_{t \geq t_0} E_p (\| \bm{\xi}_i \|^2) \leq  \xi_i^{\ast} $
for a positive and bounded constant $ \xi_i^{\ast} $. Furthermore, $ \bm{g}_{i,q}(\bar{\bm{x}}_{i,q}): \mathbb{R}^q \rightarrow \mathbb{R}^m $ is a known and differentiable function. 

\begin{assumption}\upshape \label{AS1}
	The leader trajectory $ z_{r}$ is $ n $-times continuously differentiable and there exists a positive constants $ c_{z} $ such that $ | z_{r}^{(l)} |  \leq c_{z}, $ for all $ l=1,\ldots,n $ where $ z_{r}^{(l)} $ stands for the $ l $-th derivative of $ z_{r} $.
\end{assumption}

\subsubsection{Distributed Virtual Systems}
Suppose that a virtual system is defined locally for each agent $ i $, which evolves according to the following ordinary differential equation
\begin{align} \label{trajectory}
\dot{\bm{\eta}}_{i} = \bm{A} \bm{\eta}_{i} + \bm{B} \bm{K} \bm{\zeta}_{i},
\end{align}
where $ \bm{\eta}_{i} \in \mathbb{R}^n$ with $  \bm{\eta}_{i} = [\eta_{i,1},\thinspace \eta_{i,2}, \ldots, \thinspace \eta_{i,n}]^{\rm {T}} $ is the objective trajectory of the virtual system. $ \bm{\zeta}_{i} \in \mathbb{R}^n $ is an intermittent (with respect to the communication modes) distributed virtual signal which aims at forcing the system (\ref{trajectory}) to achieve the tracking consensus with respect to the leader's information. $ \bm{K} \in \mathbb{R}^{1 \times n} $ is a control gain vector to be chosen according to the design requirements. The pair $ (\bm{A}, \bm{B}) $ is stabilizable. The matrix $ \bm{A} \in \mathbb{R}^{n \times n} $ and the vector $ \bm{B}  \in \mathbb{R}^n $ are defined as
\begin{align*} 
\bm{A} = \kern-0.3em \left[\begin{matrix} 
\bm{0}_{(n-1) \times 1}   &    \bm{I}_{(n-1) \times (n-1) }  \\
0  &  \bm{0}_{1\times (n-1) }
\end{matrix} \right],~~~~~\bm{B} = [\bm{0}_{1\times (n-1) }~~ \thinspace 1]^{\rm {T}}. 
\end{align*}

\begin{figure*}
	\centering
	\includegraphics[scale = 0.72]{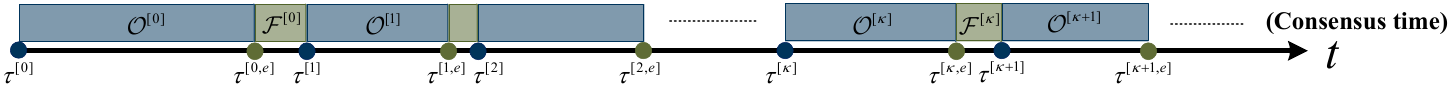}
	\captionsetup{font=footnotesize}
	\caption[]{Schematic time diagram for the intermittent communication operation modes.  	\label{INTER}}
\end{figure*}

\section{Design and Analysis} 
\label{Sec. 3}
Let $ \tau^{[0]} , \tau^{[1]} , \ldots, \tau^{[\kappa]} , \tau^{[\kappa+1]} , \ldots $ denote the time sequence at which the communication network switches to the ON mode. In our problem formulation, the first communicating ON mode instant is set to $ \tau^{[0]} =0 $. Subsequent to each ON mode interval, the networked system experiences an OFF mode of a particular duration. Let $ \tau^{[0,e]} , \tau^{[1,e]} , \ldots, \tau^{[\kappa,e]} , \tau^{[\kappa+1,e]} , \ldots $ denote the time sequence at which the communication network switches to the OFF mode. The overall time sequence of switches in order of occurrence is as follows: $ \tau^{[0]}, \tau^{[0,e]} , \tau^{[1]},  \tau^{[1,e]} , \ldots, \tau^{[\kappa]}, \tau^{[\kappa,e]} , \ldots $. To recapitulate, the following holds true regarding the communication intermittent modes
\begin{align} \label{Intermittent}
\mathscr{M}:
\begin{cases}
\mathcal{O}^{[\kappa]}:~ t \in [\tau^{[\kappa]}, \tau^{[\kappa,e]})~~ &\text{ON Mode} ,\\
\mathcal{F}^{[\kappa]}:~ t \in [\tau^{[\kappa,e]}, \tau^{[\kappa+1]})~~ &\text{OFF Mode}    ,
\end{cases}
\end{align}
for $ \kappa \in \mathbb{N} $. The total duration of two consecutive On and OFF mode intervals is denoted by $ {T}^{[\kappa]} = \tau^{[\kappa+1]} - \tau^{[\kappa]} $ which satisfies $ 0 < {T}^{[\kappa]} \leq {T}$ for a positive constant $ {T} $. The duration of a particular ON mode interval for communication network is also defined as $ \mathscr{G}^{[\kappa]} = \tau^{[\kappa,e]} - \tau^{[\kappa]}$. Easily, the duration of a particular OFF mode interval for communication network is derived as $ \mathscr{C}^{[\kappa]} = {T}^{[\kappa]} - \mathscr{G}^{[\kappa]}$. To better illustrate the time notation regarding the intermittent operation modes, a schematic diagram is provided in Fig. \ref{INTER}.


Design the intermittent distributed virtual control signal $ \bm{\zeta}_{i}(t) $ as
\begin{align} \label{Intermittentvirtualcontrol}
\bm{\zeta}_{i} =
\begin{cases}
\sum_{j \in {N}_i}  a_{i,j} (\bm{\eta}_{j} - \bm{\eta}_{i} ) + b_i (\bar{\bm{z}}_{r} - \bm{\eta}_{i}), & t \in \mathcal{O}^{[\kappa]},\\
\bm{0}_{n \times 1}, & t \in  \mathcal{F}^{[\kappa]}   ,
\end{cases}
\end{align}
where $ \bar{\bm{z}}_{r} =  [ z_{r},\thinspace z_{r}^{\prime},\thinspace \ldots, z_{r}^{(n-1)}]^{\rm {T}} $.

Let $ \bm{\varsigma} = [\bm{\varsigma}_{1}^{\rm {T}},\thinspace \ldots,\thinspace \bm{\varsigma}_{N}^{\rm {T}}]^{\rm {T}} $ and $  \bm{\eta} = \big[\bm{\eta}_{1}^{\rm {T}},\thinspace \ldots,\thinspace \bm{\eta}_{N}^{\rm {T}} \big]^{\rm {T}} $, where $ \bm{\varsigma}_{i} = \bm{\eta}_{i} - \bar{\bm{z}}_{r} $.
A Lyapunov function candidate is then chosen as 
$$ \mathcal{V}_{e} = \bm{\varsigma}^{\rm {T}} ( \mathbf{I}_N \otimes \bm{P} ) \bm{\varsigma} $$
in which $ \bm{P} $ is a positive definite matrix that satisfies the following inequalities
\begin{align} \small \label{Ricatti}
\begin{cases}
\bm{A}^{\rm {T}} \bm{P} + \bm{P} \bm{A} - 2 \bm{P} \bm{B} \bm{B}^{\rm {T}} \bm{P} + c_{1} \bm{I} <  0, \\
\bm{A}^{\rm {T}} \bm{P} + \bm{P} \bm{A} - c_{3} \bm{P} <  0, 
\end{cases}
\end{align}
where $ c_{1} $ and $ c_{3} $ are some positive design constants. According to \eqref{Intermittentvirtualcontrol}, for all $ t \in \mathcal{O}^{[\kappa]} $, the time derivative of the Lyapunov function candidate $ \mathcal{V}_{e} $ gives
\begin{align} \label{Lya1dot}
\dot{\mathcal{V}}_{e} &= \bm{\varsigma}^{\rm {T}} (  \mathbf{I}_N \otimes ( \bm{A}^{\rm {T}} \bm{P} + \bm{P} \bm{A}) - (\bm{\mathcal{L}}+\bm{\mathcal{B}})\otimes 2\bm{P} \bm{B} \bm{K}  ) \bm{\varsigma} \nonumber \\
&- 2 \bm{\varsigma}^{\rm {T}} ( \mathbf{1}_N \otimes z_{r}^{(n)} \bm{P} \bm{B} ).
\end{align}
The control gain vector is designed as 
$ \bm{K} = c_{0} \bm{B}^{\rm {T}} \bm{P} $ 
where the positive integer $ c_{0} $ is chosen such that 
$ c_{0} \lambda_{\min}(\bm{\mathcal{L}}+\bm{\mathcal{B}}) \geq 1 $.
The time derivative of the Lyapunov function candidate $ \dot{\mathcal{V}}_{e} $ satisfies 
$$ \dot{\mathcal{V}}_{e} \leq - c_{\alpha,1}  \bm{\varsigma}^{\rm {T}} \bm{\varsigma} + c_{\beta} $$ 
with the constants    
$$ c_{\alpha,1} = c_{1} -\frac{ c_{z} N \|\bm{P} \|}{c_{2}}$$
and 
$$ c_{\beta} = c_{2} c_{z} \|\bm{P}\|$$ 
and positive constant $ c_{2} $. 
Hence, 
$$ \dot{\mathcal{V}}_{e}(t) \leq - \delta_{\alpha} \mathcal{V}_{e}(t) + c_{\beta} $$
with 
$$ \delta_{\alpha} = \frac{c_{\alpha,1} }{\|\bm{P}\|}. $$

\noindent
 For $ t\in \mathcal{O}^{[\kappa]} $, it holds that
\begin{align} \label{Lyadfffyyy}
\mathcal{V}_{e}(t)  \leq  \mathcal{V}_{e}(\tau^{[\kappa]}) e^{-\delta_{\alpha}(t-\tau^{[\kappa]})}  + \frac{c_{\beta}}{\delta_{\alpha}}(1-e^{-\delta_{\alpha} (t-\tau^{[\kappa]})}).
\end{align}

\noindent
 For all $ t \in \mathcal{F}^{[\kappa]} $, let 
 $$ \delta_{\beta} = c_{3} + \frac{ c_{z} N }{c_{2}}.$$  
 The time derivative of the Lyapunov function candidate gives 
 $$ \dot{\mathcal{V}}_{e}(t) \leq  \delta_{\beta} \mathcal{V}_{e}(t) + c_{\beta}.$$

 \noindent
Consequently, it holds that
\begin{align} \label{Lyadotttt5extra}
\mathcal{V}_{e}(t)  \leq \mathcal{V}_{e}(\tau^{[\kappa,e]}) e^{\delta_{\beta}(t-\tau^{[\kappa,e]})}  + \frac{c_{\beta}}{\delta_{\beta}}(e^{\delta_{\beta} (t-\tau^{[\kappa,e]})} - 1).
\end{align}

\noindent
Let 
$$ \Lambda^{[\kappa]} =  \delta_{\alpha}(\tau^{[\kappa,e]}-\tau^{[\kappa]}) - \delta_{\beta}(\tau^{[\kappa+1]} - \tau^{[\kappa,e]}),$$
and 
$$ \varkappa =  (\frac{c_{\beta}}{\delta_{\beta}} + \frac{c_{\beta}}{\delta_{\alpha}}) e^{\delta_{\beta} {T}} - \frac{c_{\beta}}{\delta_{\beta}}.$$
Therefore, there exists a $ {T} $ for which $ \varkappa > 0 $.


\begin{thm} \label{THEEERM1} \upshape 
Under Assumption 1, consider the auxiliary trajectory of $ i $-th agent in \eqref{trajectory}. Suppose that the data is intermittently transmitted among the agents according to the intermittent rule in \eqref{Intermittentvirtualcontrol}. Then, design the control gain $ \bm{K} = c_{0} \bm{B}^{\rm {T}} \bm{P}$ for the objective trajectory. If for all  $ \kappa \in \mathbb{N} $
\begin{align} \label{Condi}
 (\delta_{\alpha} + \delta_{\beta})\mathscr{G}^{[\kappa]} -  \delta_{\beta} {T}^{[\kappa]} > 0,
\end{align}
holds true, then the communication network is resilient to be switched to the OFF mode with a maximum duration of 
$$ \mathscr{C}^{[\kappa]} \approx  (\mathscr{G}^{[\kappa]}) (\frac{\delta_{\alpha} + \delta_{\beta}}{\delta_{\beta}} - 1) \times 100 \%.$$
Also, it holds that 
$ \lim\nolimits_{t \rightarrow \infty} | \eta_{i} - z_{r} | \leq \varepsilon^{[e]}$,
for a positive constant $ \varepsilon^{[e]} $.
\end{thm}

\begin{proof}
By recursion, the following is derived 
\begin{align}  \label{lyfff}
{\mathcal{V}}_{e}(\tau^{[\kappa+1]}) &\leq [{\mathcal{V}}_{e}(\tau^{[\kappa]}) - \frac{c_{\beta}}{\delta_{\alpha}}] e^{-\Lambda^{[\kappa]}} \nonumber \\
&~~~~~~~~~~~+ (\frac{c_{\beta}}{\delta_{\beta}} + \frac{c_{\beta}}{\delta_{\alpha}}) e^{\delta_{\beta}(\tau^{[\kappa+1]} - \tau^{[\kappa,e]})} - \frac{c_{\beta}}{\delta_{\beta}} \nonumber \\
& \leq [{\mathcal{V}}_{e}(\tau^{[\kappa]}) - \frac{c_{\beta}}{\delta_{\alpha}}] e^{-\Lambda^{[\kappa]}} + \varkappa \nonumber \\
& \leq {\mathcal{V}}_{e}(t_0) e^{-\sum_{l=0}^{\kappa} \Lambda^{[l]} } - \frac{c_{\beta}}{\delta_{\alpha}}[e^{-\sum_{l=0}^{\kappa} \Lambda^{[l]} } \nonumber \\
& +e^{-\sum_{l=1}^{\kappa} \Lambda^{[l]} } + \ldots + e^{-\sum_{l=\kappa}^{\kappa} \Lambda^{[l]} }] \nonumber \\
& + \varkappa [e^{-\sum_{l=1}^{\kappa} \Lambda^{[l]} } + e^{-\sum_{l=2}^{\kappa} \Lambda^{[l]} } + \ldots \nonumber \\
&  + e^{-\sum_{l=\kappa}^{\kappa} \Lambda^{[l]} }] + \varkappa \nonumber \\
& \leq {\mathcal{V}}_{e}(t_0) e^{-\sum_{l=0}^{\kappa} \Lambda^{[l]} } + \varkappa [e^{-\sum_{l=1}^{\kappa} \Lambda^{[l]} } \nonumber \\
&  + e^{-\sum_{l=2}^{\kappa} \Lambda^{[l]} } + \ldots  + e^{-\sum_{l=\kappa}^{\kappa} \Lambda^{[l]} }] + \varkappa.
\end{align}
Let $ \Lambda = \min\limits_{0 \leq l \leq \kappa} \Lambda^{[l]}$. Then, 
\begin{align}  \label{lyfffextra}
{\mathcal{V}}_{e}(\tau^{[\kappa+1]}) &\leq {\mathcal{V}}_{e}(t_0) e^{-(\kappa+1) \Lambda }  + \varkappa [e^{-\kappa \Lambda } + \ldots + e^{-\Lambda}] + \varkappa \nonumber \\
& \leq [{\mathcal{V}}_{e}(t_0) + \frac{\varkappa e^{\Lambda}}{1 -  e^{\Lambda}}] e^{-(\kappa+1) \Lambda } \nonumber \\
&~~~-\frac{\varkappa }{1 -  e^{\Lambda}} + \varkappa.
\end{align}
For all $ t > 0 $, there exists $ \iota \in \mathbb{N} $ such that 
$ t \in [\tau^{[\iota+1]}, \tau^{[\iota+2]}) $. 
Since, $ t < (\iota+2){T}$, it also holds that 
$$ -(\iota+1) \Lambda \leq -\frac{\Lambda}{{T}}t - \Lambda.$$
For $ t \in [\tau^{[\iota+1]}, \tau^{[\iota+1,e]}) $, it holds that 
\begin{align}  \label{lyainterzx}
{\mathcal{V}}_{e}(t) &\leq \mathcal{V}_{e}(\tau^{[\iota+1]}) e^{-\delta_{\alpha}(t - \tau^{[\iota+1]})} +  \frac{c_{\beta}}{\delta_{\alpha}} (1 - e^{-\delta_{\alpha}(t - \tau^{[\iota+1]})})\nonumber \\
& \leq \mathcal{V}_{e}(\tau^{[\iota+1]}) + \frac{c_{\beta}}{\delta_{\alpha}} \nonumber \\
& \leq [{\mathcal{V}}_{e}(t_0) + \frac{\varkappa e^{\Lambda}}{1 -  e^{\Lambda}}] e^{-(\iota+1) \Lambda } -\frac{\varkappa }{1 -  e^{\Lambda}} + \varkappa + \frac{c_{\beta}}{\delta_{\alpha}} \nonumber \\
& \leq \pi_{1} e^{-\frac{\Lambda}{{T}}t}  + \varepsilon_{1},
\end{align}
where 
$$ \pi_{1} =  {\mathcal{V}}_{e}(t_0) e^{- \Lambda} + \frac{\varkappa }{1 - e^{\Lambda}}$$ 
and 
$$ \varepsilon_{1} = -\frac{\varkappa }{1 -  e^{\Lambda}} + \varkappa + \frac{c_{\beta}}{\delta_{\alpha}}.$$ 
Similarly for $ t \in [\tau^{[\iota+1,e]}, \tau^{[\iota+2]}) $, it holds that
\begin{align}  \label{lyaintercccxcxextra}
{\mathcal{V}}_{e}(t) \leq \pi_{2} e^{-\frac{\Lambda}{{T}}t}  + \varepsilon_{2},
\end{align}
where 
$$ \pi_{2} = e^{\delta_{\beta}{T}} [{\mathcal{V}}_{e}(t_0) e^{- \Lambda} + \frac{\varkappa }{1 - e^{\Lambda}}]$$
and 
$$ \varepsilon_{2} =  e^{\delta_{\beta}{T}} [-\frac{\varkappa }{1 -  e^{\Lambda}} + \varkappa + \frac{c_{\beta}}{\delta_{\alpha}}] + \varkappa.$$
\\
Let $ \pi^{\star} = \max \{\pi_{1}, \pi_{2} \}$ and $ \varepsilon^{\star} = \max \{\varepsilon_{1}, \varepsilon_{2} \} $. Then, 
\begin{align} \label{VFinal}
{\mathcal{V}}_{e}(t) \leq  \pi^{\star} e^{-\frac{\Lambda}{{T}}t} +  \varepsilon^{\star},
\end{align}
for all $ t > 0 $. If \eqref{Condi} holds true, then $ \Lambda > 0 $ and 
$$ \max\{\mathscr{C}^{[\kappa]}\} \approx  (\mathscr{G}^{[\kappa]}) (\frac{\delta_{\alpha} + \delta_{\beta}}{\delta_{\beta}} - 1). $$ 
Hence, the closed-loop systems error $ \bm{\varsigma}_{i} $ exponentially converge to $ \varepsilon^{\star} $. Accordingly, the state vector of each agent $ i $ exponentially converges to a bounded region around $ \bar{\bm{z}}_{r} $ as time tends to infinity, i.e., 
$$ \lim\nolimits_{t \rightarrow \infty} | \eta_{i} - z_{r} | \leq \varepsilon^{[e]}.$$ 
This completes the proof. 
\end{proof}

To continue, the following closed-loop system errors are defined 
\begin{align} \label{error1nonlinear}
e_{i,1} = z_{i} - \eta_{i,1}, ~~~~ e_{i,q} = x_{i,q} - \alpha_{i,q-1}
\end{align}
where $ q = 2, \ldots, n $ and $ \alpha_{i,q-1} $ is the virtual input to be designed later. We will derive the control laws and adaptive rules for the follower agents in what follows.
\\
Let $ \bar{f}_{i,1} = f_{i,1}(\bar{\bm{x}}_{i,1}) - \eta_{i,2} $ and $ \bar{f}_{i,q} = f_{i,q}(\bar{\bm{x}}_{i,q}) - \dot{\alpha}_{i,q-1} $ for $ q = 2, \ldots, n $. The neural networks (NNs)\cite{12} approach, is in most cases utilized locally to approximate unknown functions of the form $ \bar{f}_{i,q} $ for $ q = 1,\ldots,n $. Hence, for positive constants $\epsilon_{i,q}  $, it follows that $ \bar{f}_{i,q} = \bm{\theta}_{i,q}^{\rm T} \bm{\varphi}_{i,q} + \epsilon_{i,q}  $. $ \bm{\theta}_{i,q}$ is an unknown vector. The activation vectors of Gaussian basis functions are denoted by $ \bm{\varphi}_{i,q} $. Additionally, one has $|\epsilon_{i,q}|^2 \leq \epsilon_{i,q}^{\ast}$, for some unknown positive bounded parameters $ \epsilon_{i,q}^{\ast} $. Since $ \bm{\theta}_{i,q}$ is an unknown constant vector, the approximation errors can be defined as 
$$ \tilde{\bm{\theta}}_{i,q}  = \bm{\theta}_{i,q}  -  \hat{\bm{\theta}}_{i,q}  $$ 
where $ \hat{\bm{\theta}}_{i,q} $ stands for the approximated value of $ \bm{\theta}_{i,q}$. 
By applying Young's inequality one gets 
$$ e_{i,q} e_{i,q+1} \leq \frac{\rho_{i,q}^2}{2} e_{i,q}^2  + \frac{1}{2 \rho_{i,q}^2} e_{i,q+1}^2 $$
for $ q = 1, \ldots, n-1 $, and the followings hold true for $ q = 1, \ldots, n $
\begin{align*} 
\begin{cases}
e_{i,q} \epsilon_{i,q} \leq \frac{\rho_{i,q}^2}{2} e_{i,q}^2 + \frac{1}{2 \rho_{i,q}^2} \epsilon_{i,q}^{\ast},  \\
e_{i,q} \bm{g}_{i,q}^{\rm {T}}(\bar{\bm{x}}_{i,q}) \bm{\xi}_i \leq  \frac{\rho_{i,q}^2}{2} e_{i,q}^2 \| \bm{g}_{i,q} \|^2 + \frac{1}{2 \rho_{i,q}^2} \| \bm{\xi}_i \|^2, 
\end{cases}
\end{align*}
with positive design constants $ \rho_{i,q},~q = 1, \ldots, n $.
\\
\indent 
The actual control input $ u_i $ is designed as follows 
\begin{align} 
u_i &= -(\mathcal{K}_{i,n} + \frac{\rho_{i,n}^2}{2} \| \bm{g}_{i,n} \|^2)  e_{i,n} - \hat{\bm{\theta}}_{i,n}^{\rm T} \bm{\varphi}_{i,n}, \label{actual} 
\end{align}
and the $ q $-th virtual input $ \alpha_{i,q} $ for $ q = 1, \ldots, n-1 $, and the $ q $-th adaptive rule for $ q = 1, \ldots, n $ are 
\begin{align} 
\alpha_{i,q} &= -(\mathcal{K}_{i,q} + \frac{\rho_{i,q}^2}{2} \| \bm{g}_{i,q} \|^2)  e_{i,q} - \hat{\bm{\theta}}_{i,q}^{\rm T} \bm{\varphi}_{i,q}, \label{Virtual} \\
\dot{\hat{\bm{\theta}}}_{i,q}  &= \gamma_{i,q}( {e}_{i,q}\bm{\varphi}_{i,q} - \sigma_{i,q} \hat{\bm{\theta}}_{i,q}  ), \label{Adaptiveq}
\end{align}
where $ \sigma_{i,q} $ is a $ \sigma $-modification parameter and $ \mathcal{K}_{i,q} $ is a positive gain to be chosen according to the control objectives. The design procedure is now complete and the following Theorem is provided to summarize the results.

\begin{thm} \label{THEEERM2} \upshape 
Consider a group of uncertain nonlinear MASs defined in (\ref{dyn}). Design the virtual signals and adaptive rules in \eqref{Virtual}, \eqref{Adaptiveq}. Furthermore, consider that the actuator updates are done according to \eqref{actual}. Then, the tracking errors are NSpS-P with a unique global solution. Furthermore, 
$$ \lim\nolimits_{t \rightarrow \infty} E_p (|z_{i} - z_{r}|) < \varepsilon_{i,v} $$
for all $ i \in V $ and a positive constant $ \varepsilon_{i,v} $ i.e., the tracking consensus performance can be achieved. 
\end{thm}

\begin{proof}
    Let 
    $ \Delta_{i,q} = \mathcal{K}_{i,q} - \rho_{i,q}^2 $
    for $ q = 1, \ldots, n-1 $
    and 
    $ \Delta_{i,n} = \mathcal{K}_{i,n} - 0.5 \rho_{i,n}^2 $. 
    By applying Young's inequality, it follows that 
    $$ \sigma_{i,q} \tilde{\bm{\theta}}_{i,q}^{\rm T} \bm{\theta}_{i,q} \leq \frac{1}{2}  \sigma_{i,q} \tilde{\bm{\theta}}_{i,q}^{\rm T} \tilde{\bm{\theta}}_{i,q} + \frac{1}{2}  \sigma_{i,q} \| \bm{\theta}_{i,q} \|^2. $$
    A Lyapunov function candidate is chosen as 
\begin{align} 
\mathcal{V}_{i} = \sum_{q=1}^{n} \frac{1}{2} e_{i,q}^{2} + \sum_{q=1}^{n} \frac{1}{2 \gamma_{i,q} } \tilde{\bm{\theta}}_{i,q}^{\rm T} \tilde{\bm{\theta}}_{i,q},
\end{align}
where $ \gamma_{i,q} $ is a positive design parameter. Then, the time derivative of $ \mathcal{V}_{i} $ gives
\begin{align} \label{Lyapunovtotfooo}
\dot{\mathcal{V}}_{i} \leq &-\Delta_{i,1} e_{i,1}^2 -\sum_{q=2}^{n} (\Delta_{i,q} - \frac{1}{2 \rho_{i,q-1}^2}) e_{i,q}^2 \nonumber \\
& - \frac{1}{2} \sum_{q=1}^{n} \sigma_{i,q} \tilde{\bm{\theta}}_{i,q}^{\rm T} \tilde{\bm{\theta}}_{i,q} +  l_{i}^{[a]} \| \bm{\xi}_i \|^2 + l_{i}^{[b]},
\end{align}
where $ l_{i}^{[a]} = \sum_{q=1}^{n} \frac{1}{2 \rho_{i,q}^2} $ and 
$$ l_{i}^{[b]} = \sum_{q=1}^{n} \frac{1}{2 \rho_{i,q}^2} \epsilon_{i,q}^{\ast} + \sum_{q=1}^{n} \frac{1}{2}  \sigma_{i,q} \| \bm{\theta}_{i,q} \|^2.$$
Let 
\begin{align*}\small
c_{i,\gamma} = \min \bigg\{ &2\Delta_{i,1},  \big[ 2(\Delta_{i,q} - \frac{1}{2 \rho_{i,q-1}^2}) \big]_{q=2,\ldots,n,}\\
&\big[ \frac{\gamma_{i,q}}{\sigma_{i,q}}  \big]_{q=1,\ldots,n.} \bigg\}.
\end{align*}
Hence, 
\begin{align}\label{Vtotalfffoooo}
\dot{\mathcal{V}}_{i,o} \leq  -c_{i,\gamma} \mathcal{V}_{i,o} + l_{i}^{[a]} \| \bm{\xi}_i \|^2 + l_{i}^{[b]}.
\end{align}
Therefore, by defining 
$ \varpi_{i} = l_{i}^{[a]}  \xi_i^{\ast} + l_{i}^{[b]} $ 
and applying Lemma 1, it holds that 
$$ \lim\nolimits_{t \rightarrow \infty} E_p (\mathcal{V}_{i,o}) \leq \varpi_{i} / c_{i,\gamma}.  $$ 

Hence, the error signals $ e_{i,q} $, and $ \tilde{\bm{\theta}}_{i,q} $ for $ q = 1, \ldots,n $ are all NSpS-P with a unique global solution according to Lemma 1. Accordingly, 
$ \lim\nolimits_{t \rightarrow \infty} E_p (|z_{i} - \eta_{i,1}|) < \varepsilon_{i,x} $ 
for a positive constant $ \varepsilon_{i,x} $. From \eqref{VFinal} it follows that
$ \lim\nolimits_{t \rightarrow \infty} | \eta_{i} - z_{r} | \leq \varepsilon^{[e]}$.
Therefore, 
$$ \lim\nolimits_{t \rightarrow \infty} E_p (|z_{i} - z_{r}|) < \varepsilon_{i,v} $$ 
and the tracking consensus performance can ultimately be achieved.
This completes the proof.
\end{proof}

\begin{figure}[t]
	\centering
	\includegraphics[scale = 1.5]{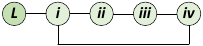} 
	\captionsetup{font=footnotesize}
	\caption[]{The communication graph. 	\label{Graph}}
\end{figure}


 \begin{figure}
	\centering
	\includegraphics[scale = 0.2]{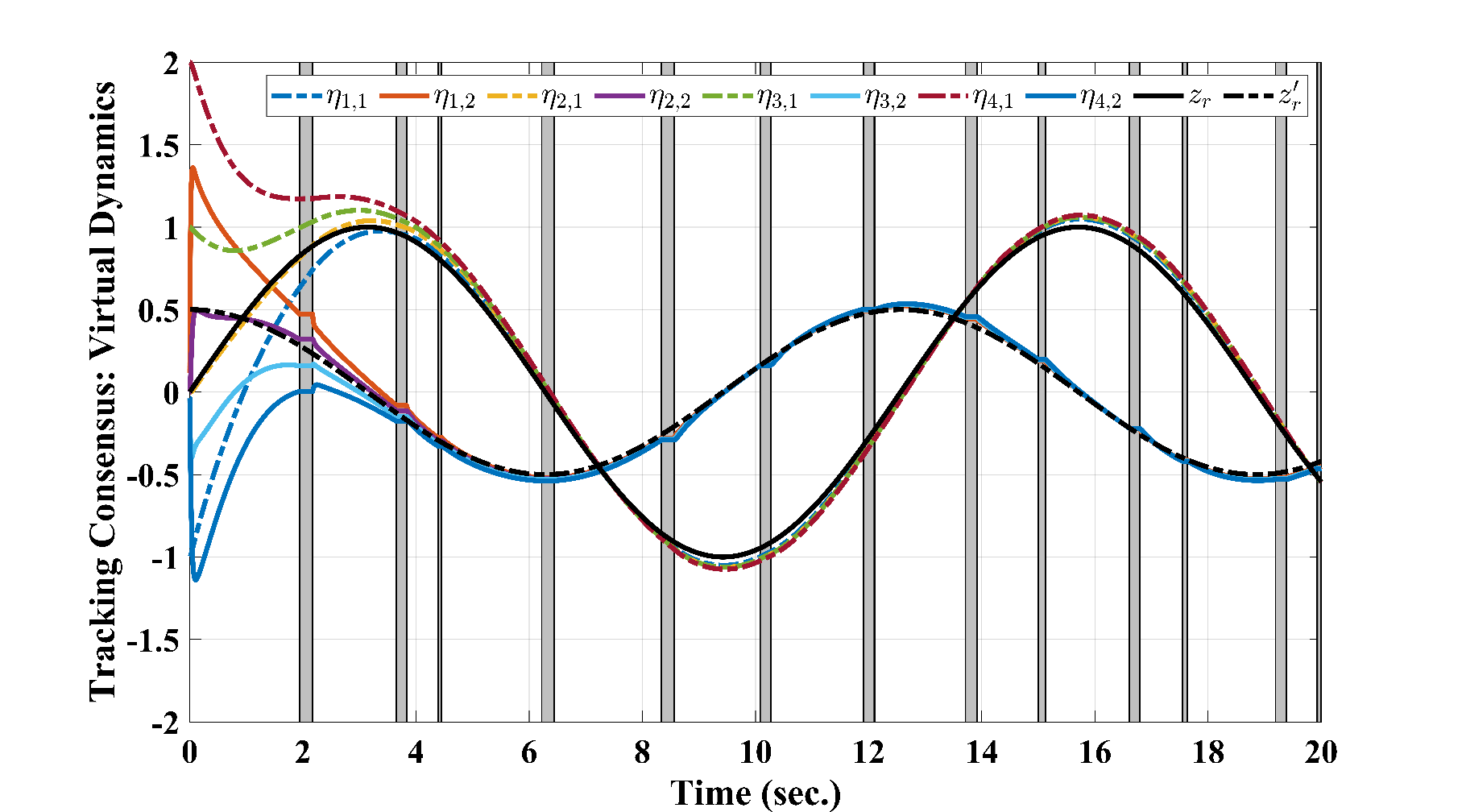} 
	\captionsetup{font=footnotesize}
	\caption[]{Consensus trajectory tracking: $ \eta_{i,q} $ with respect to $ z_{r} $ and $z_{r}^{\prime}$.  	\label{Tracking}}
\end{figure}

 \begin{figure}
	\centering
	\includegraphics[scale = 0.2]{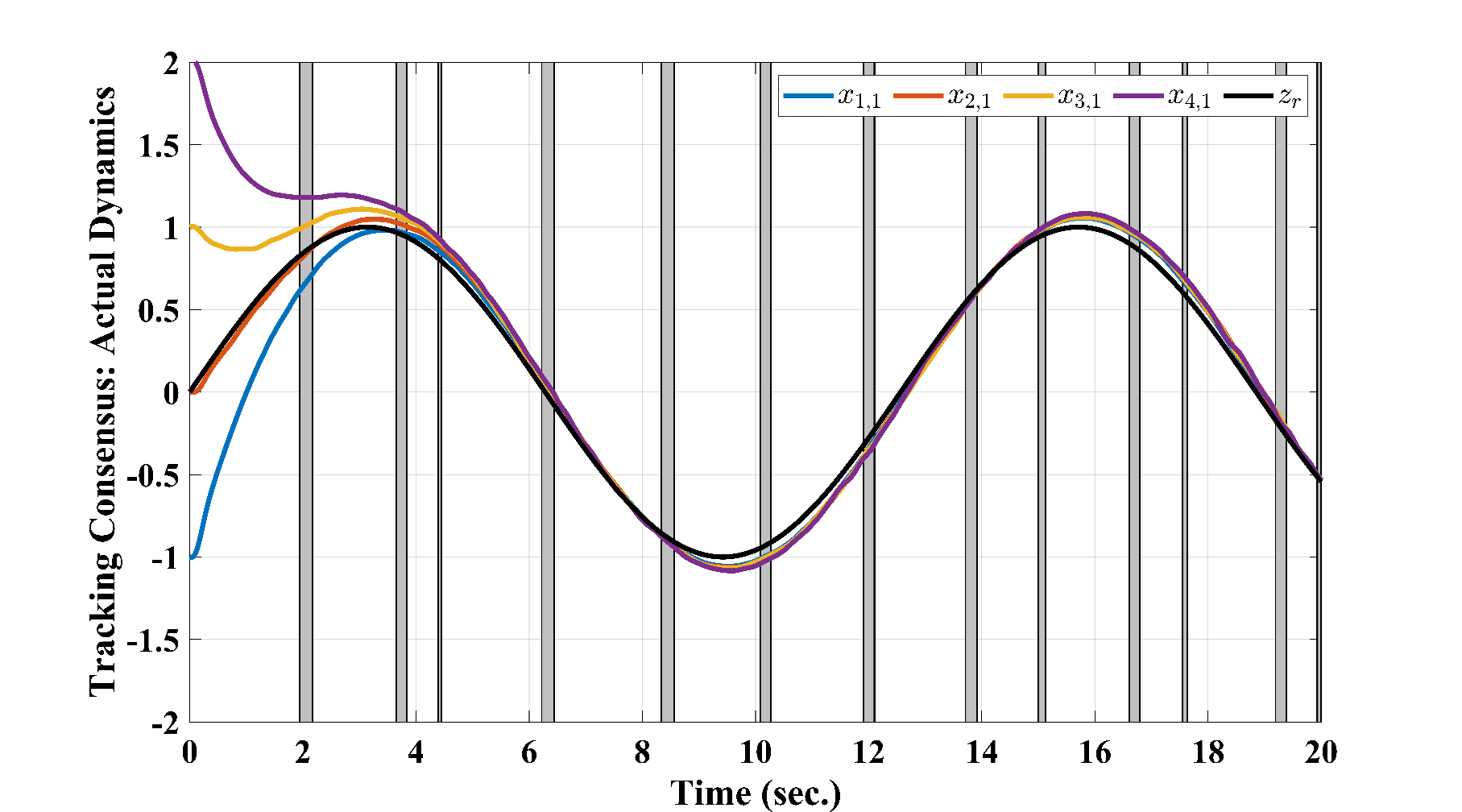}
	\captionsetup{font=footnotesize}
	\caption[]{Consensus trajectory tracking: $ x_{i,1} $ with respect to $ z_{r} $.  	\label{Tracking2}}
\end{figure}


\section{Simulation results}
\label{Sec. 4}
 In this section, the simulation results are reported subsequent to applying the proposed method in this paper to a network of $ N = 4 $ uncertain nonlinear second-order agents (i.e., $ n=2 $) defined in (\ref{dyn}). The agents are initially stationary positioned randomly in $ [-2,2] $. 
 \\
 The unknown smooth heterogeneous nonlinear functions are defined as 
 $ f_{i,1}(\bar{\bm{x}}_{i,1}) = 0.5ix_{i,1}\sin(x_{i,1})\cos(x_{i,1})$ and 
 $ f_{i,2}(\bar{\bm{x}}_{i,2}) = 0.9ix_{i,1}\sin(x_{i,2})\cos(0.3x_{i,1})$.
 The virtual leader signal is 
 $ z_{r} = \sin(0.5t)$
 and therefore $ c_{z} = 1 $. The random disturbances coefficient functions with $ m=1 $ are 
 $ {g}_{i,1}(\bar{\bm{x}}_{i,1}) = 0.5x_{i,1}\sin(x_{i,1})$ and 
 $ {g}_{i,2}(\bar{\bm{x}}_{i,2}) = 0.5x_{i,1}\sin(x_{i,1})\cos(x_{i,2}) $.
 In the simulation results, the random disturbances $ \xi_{i} $ are produced by 
 \begin{align} \label{randomdynamic}
 \dot{\xi}_{i}(t) = \frac{1}{ \mathcal{X}} [-\xi_{i}(t) + w_i(t)], 
 \end{align}
 in which $  w_i(t) $ is generated by band-limited white noise block with noise power $ A = 1 $ and correlation time $ t_c = 0.1 $.  The number of seed for $  w_i(t) $ is $ [23341] $, $ [34243] $, $ [23343] $, and $ [34241] $, respectively for $ i = 1,\ldots,4 $.
 \\
 The communication graph for the considered system is depicted in Fig. \ref{Graph}. From the graph topology, 
 $ \lambda_{\min}(\bm{\mathcal{L}}+\bm{\mathcal{B}}) = 0.1981 $.
 With $ c_{0} = 6$, it holds that 
 $ c_{0} \lambda_{\min}(\bm{\mathcal{L}}+\bm{\mathcal{B}}) \geq 1 $. A feasible solution to \eqref{Ricatti} gives 
 \begin{align*} 
 \bm{P} = \kern-0.3em \left[\begin{matrix} 
 22.9454  &  3.1623 \\
 3.1623   & 3.6280    
 \end{matrix} \right],
 \end{align*}
 with $ c_{1} = 20$, $ c_{2} = 10 $ and $ c_{3} = 3$. 
 The control gain vector is chosen as 
 $ \bm{K} = [18.9737 ~~ \thinspace  21.7679 ] $.
 Furthermore, $ \delta_{\alpha} = 0.4529 $ and $ \delta_{\beta} = 3.4 $. Additionally, from \eqref{Condi}, 
 $$ \max\{\mathscr{C}^{[\kappa]}\} \approx  (\mathscr{G}^{[\kappa]}) \times 13.2 \%. $$ 
 In the simulation results, we choose $ \mathscr{C}^{[\kappa]} = 0.9 \max\{\mathscr{C}^{[\kappa]}\}$ and the duration of a particular ON mode interval for communication network is randomly selected from $ [0.5,2] $. The other design parameters are also selected as $ \mathcal{K}_{i,q} = 15 $, $ \sigma_{i,q} = 0.5 $, and $ \gamma_{i,q} = 10 $ for $ q = 1, \ldots, n $. 
 \\
 The simulation is performed for a total time period of $ 20 $ seconds with the sampling time-step one milliseconds. 
The tracking consensus of virtual and real systems are respectively depicted in Fig. \ref{Tracking} and Fig. \ref{Tracking2}, in which the OFF mode intervals are highlighted in gray color. Otherwise, the system is performing the tasks in ON mode intervals. 
From the simulation experiments, it can be seen that the tracking performance is achieved while the communications among neighboring agents are intermittent. The effects of random disturbance are also compensated with the proposed approach in this paper.

	\section{Conclusion}
	\label{Sec. 5}
	An intermittent communication mechanism for the tracking consensus of high-order nonlinear MASs surrounded by random disturbances was addressed in this paper. The resiliency level to the failure of physical devices or unreliability of communication channels was analyzed by introducing a linear auxiliary trajectory of the system. The closed-loop networked system signals were proved to be NSpS-P. It has been justified that each agent follows the trajectory of the corresponding local auxiliary virtual system practically in probability.

	\bibliographystyle {ieeetr}
	\small\bibliography {Refrences}
	
\end{document}